\documentclass[12pt]{article}
\usepackage{amsfonts}
\usepackage{amsmath}
\usepackage{setspace}
\usepackage{longtable}
\usepackage{graphicx,amsthm,verbatim,bm,color}
\usepackage[sort&compress]{natbib}
\usepackage{hyperref}

\newtheorem{theorem}{Theorem}

\newcommand{\til}       {\mbox{$\tilde{\theta}_{i}$}}
\newcommand{\tilb}       {\mbox{$\tilde{\theta}_{i}^B$}}

\newcommand{\thib}       {\mbox{$\hat{\theta}_{i}^B$}}
\newcommand{\thijb}       {\mbox{$\hat{\theta}_{ij}^B$}}

\newcommand{\V}       {\text{Var}}
\newcommand{\db}       {\mbox{$\hat{\delta}_i^B$}}
\newcommand{\dr}       {\mbox{$\hat{\delta}_i^R$}}
\newcommand{\di}       {\mbox{$\hat{\delta}_i^I$}}

\newcommand{\dn}       {\mbox{$\hat{\delta}_i^n$}}
\newcommand{\dc}       {\mbox{$\hat{\delta}_i^C$}}

\newcommand{\de}       {\mbox{$\delta_i$}}

\newcommand{\tij}       {\mbox{$\theta_{ij}$}}
\newcommand{\htijb}       {\mbox{$\hat{\theta}_{ij}^B$}}

\newcommand{\thbariw}     {\mbox{$\bar{\hat{\theta}}_{iw}^{B}$}}
\newcommand{\tbw}       {\mbox{$\bar{\theta}_{w}$}}
\newcommand{\tbiw}       {\mbox{$\bar{\theta}_{iw}$}}
\newcommand{\tiw}         {\mbox{$\bar{\hat{\theta}}_{iw}$}}
\newcommand{\thw}       {\mbox{$\bar{\hat{\theta}}_{w}^B$}}
\newcommand{\thww}       {\mbox{$\bar{\hat{\theta}}_{w}$}}
\newcommand{\thbiw}       {\mbox{$\bar{\hat{\theta}}_{iw}^B$}}
\newcommand{\that}       {\mbox{$\hat{\theta}$}}
\newcommand{\thij}       {\mbox{$\hat{\theta}_{ij}$}}

\newcommand{\thhi}       {\mbox{$\hat{\theta}_{i}$}}

\title{Two-stage Benchmarking as Applied\\ to Small Area Estimation}
\author{Malay Ghosh and Rebecca C. Steorts \thanks{Alphabetical order.}
 \\ University of Florida and Carnegie Mellon University}
\date{}
\begin{document}
\maketitle
\begin{center}
\vspace*{-1em}
\textbf{Abstract}
\end{center}
\noindent 
There has been recent growth in small area estimation due to the need for more precise estimation of small geographic areas, which has led to groups such as the U.S. Census Bureau, Google, and the RAND corporation utilizing small area estimation procedures. We develop novel two-stage benchmarking methodology using a single weighted squared error loss function that combines the loss at the unit level and the area level without any specific distributional assumptions. This loss is considered while benchmarking the weighted means at each level or both the weighted means and weighted variability at the unit level. Furthermore, we provide multivariate extensions for benchmarking weighted means at both levels. 
The behavior of our methods is analyzed using a complex study from the National Health Interview Survey (NHIS) from 2000, which estimates the proportion of people that do not have health insurance for many domains of an Asian subpopulation. Finally, the methodology is explored via simulated data under the proposed model. Ultimately,  three proposed benchmarked Bayes estimators do not dominate each other,
leaving much exploration for further understanding of such complex studies such as the choice of weights, optimal algorithms for efficiency, as well as extensions to multi-stage benchmarking methods. \\

keywords: Small area estimation, benchmarking, decision theory, Bayesian methodology, MCMC, complex survey, and weighting.
\newpage
\section{Introduction}
\label{intro}
There has been recent growth in small area estimation due to the need for more precise estimation of small geographic areas, which has led to groups such as the U.S. Census Bureau, Google, and the RAND corporation utilizing small area estimation procedures. Direct estimates are often design-consistent but usually have large standard errors and coefficients of variation, especially for areas with small sample sizes. In order to produce estimates for these small areas, it is necessary to borrow strength from similar areas. Accordingly, model-based estimates often differ widely from the direct estimates, especially for areas with small sample sizes. One potential difficulty with model-based estimates is that when aggregated the overall estimate for a larger geographical area may differ considerably from the corresponding direct estimate. Moreover, an overall agreement with the direct estimates at an aggregate level is very often a political necessity to convince legislators of the utility of small area estimators. The process of adjusting model-based estimates to correct this problem is known as benchmarking. One key benefit of benchmarking is protection against model misspecification.\\

Many advancements have been made in the benchmarking literature. For instance,  
\cite{wang_2008} derived a best linear unbiased predictor (BLUP) subject to the standard benchmarking constraint while minimizing a weighted squared error loss. They restricted their attention to a simple random effects model and only considered linear estimators of small area means, all in a univariate setting. 
Recently, \cite{ugarte_2009} considered the nested mixed linear model of 
\cite{battese_1988}
and derived estimators subject to the standard benchmarking constraint. 
\cite{datta_2011} provided an extension to the work of \cite{wang_2008}, and their work was extended by \cite{bell_2013} for benchmarking a weighted mean, however nonlinear estimators were not considered.  
\cite{datta_2011}~developed a general class of benchmarked Bayes estimators, including but not restricted to linear estimators, 
subject to a more general benchmarking constraint without any specific distributional assumptions. They derived a second set of estimators under a second constraint, namely benchmarking a weighted variability. These results are extended to multivariate settings. Both papers showed that many previously proposed estimators in the literature are special cases of their general Bayes estimators.\\

Currently, federal agencies such as the U.S. Census Bureau employ a two-stage benchmarking procedure where they constrain the estimates in two separate steps. For example, the U.S. Census Bureau adjusts the state-level estimates based on a certain model to the national estimate. Then they adjust the county-level estimates to the benchmarked state-level estimates using a separate model.
Here, we   
extend the methods in \cite{datta_2011} to cover two-stage benchmarking using a single model. We do not require any specific distributional assumptions, and we consider a single loss function that incorporates information at the area level and the unit level while benchmarking either a weighted mean or both a weighted mean and weighted variability. Finally, we are able to derive estimators that again can be expressed as the ordinary Bayes estimator plus a correction factor.\\

This new approach is important since the process of combining  benchmarking at both levels into a single process is optimal in many situations, as it avoids the need to carry out two separate benchmarking steps. Hence, we can find estimates quickly and efficiently. Equally as important, our model only requires  covariate information at the unit level, whereas current benchmarking techniques require covariate information at both levels, which is not always available.
Finally, we illustrate our methods using a complex dataset taken from the  National Health Interview Survey (NHIS). We also illustrate our proposed methods using a simulation study, illustrating with real data and simulation that
none of three given benchmarked estimators dominates any of the others
in terms of posterior mean squared error, giving rise to many new research questions.

\section{Two-Stage Benchmarking}
\label{sec:two_stage}
As already described, we propose new methods where we first benchmark a weighted mean at the unit level and then the area level without imposing distributional assumptions. In Section 2.3, we benchmark the weighted means as well as the weighted variability at the unit level, and in Section 2.4, provided multivariate extensions for benchmarking weighted means at both levels. 
This work is an extension of \cite{datta_2011}. Applications include sociology, public policy, health care, economics, and general types of social statistical applications. 

\subsection{Notation}
Let $\tij$ denote the true parameter of interest for the $j$th unit (or sub-area) in the $i$th area, and let $\htijb$ denote its Bayes estimator under a certain prior $(j=1,\ldots,n_i;\;i=1,\ldots,m).$ In our notation, we have $m$ small areas, and within each $i$th area, we have $n_i$~observations.
Let $w_{ij}$ denote a set of known weights, normalized so that $\sum_{j=1}^{n_i}w_{ij}=1$ for all $i,$ which could, for example, be proportional to the population size of the $j$th sub-area.
In addition, let $\tbiw = \sum_{j} w_{ij}\tij $ denote the true weighted mean for the $i$th area. For example, $\tij$ may be the true proportion of poor school-aged children in the $j$th county in the $i$th state, and $w_{ij}$ proportional to the number of   school-aged children for the $j$th county in the $i$th state.
Also, let $\eta_i$ $(i=1,\ldots,m)$ denote the weights for the $i$th area, normalized so that $\sum_i \eta_i = 1.$ 

\subsection{Benchmarking: Weighted Mean at the Unit Level and Area Level}
We first develop a method for two-stage benchmarking a weighted mean at the unit level and the area level.
We want to find estimates $\thij$ for  $\tij$ 
and $\delta_i$ for  $\tbiw$ such that $\sum_i\eta_i \delta_i = p$ and $\sum_jw_{ij}\thij 
= \delta_i$ for all $i.$ For example, $p$ might be the estimate of the national proportion of poor 
school-aged children. Throughout we will use the notations 
$\theta=(\theta_{11},\ldots,\theta_{1n_1},\ldots,\theta_{m1},\ldots,\theta_{mn_m})^T$, $\hat{\mathbf{\theta}}= 
(\hat{\theta}_{11},\ldots,\hat{\theta}_{1n_1},\ldots,
\hat{\theta}_{m1},\ldots,\hat{\theta}_{mn_m})^T$ and
${\mathbf{\delta}} = (\delta_1,\ldots,\delta_m)^T.$
Also, define 
$\thww = \sum_{ij} \eta_i w_{ij} \thij.$
Our objective is to minimize
 $E[L(\theta,\hat{\theta},\delta)\mid \text{data}]$ with respect to $\thij$ and $\de,$ where 
\begin{align}
L(\mathbf\theta,\hat{\mathbf{\theta}},\delta) &= 
\sum_i\sum_j \xi_{ij}(\thij-\tij)^2 + \sum_i\phi_i(\de-\tbiw)^2, \label{loss}
\end{align}
subject to the benchmarking constraints (i) $\sum_i\eta_i \de = p$ and (ii) 
$\sum_jw_{ij}\thij = \de$ for all~$i$. Note that the weights $\xi_{ij}$ and 
$\phi_i$ need not be the same as $w_{ij}$ and $\eta_i$ respectively. Now 
\begin{align}
E[L(\theta,\hat{\theta},\delta)|\text{data}] &=
E\left[\sum_i\sum_j \xi_{ij}(\tij - \htijb + \htijb - \thij)^2 + \right.\notag\\
&\qquad\left.\sum_i\phi_i(\tbiw - \thbiw + \thbiw - \de )^2|\text{data}\right] \notag\\
&=\sum_i \sum_j\xi_{ij}[(\thij-\htijb)^2 +\V(\tij|\text{data})]
\label{lossone}\\
&\qquad+\sum_i\phi_i[(\de - \thbiw )^2+\V(\tbiw|\text{data})].\notag
\end{align}

In view of (\ref{lossone}), minimization of $E[L(\theta, \that, \delta)|\text{data}]$ with respect to $\that$ subject to constraints (i) and (ii)
reduces to minimizing
\begin{align}
\label{newloss}
\sum_i \sum_j \xi_{ij}(\thij-\htijb)^2 
+\sum_i\phi_i\left[\sum_jw_{ij} (\thij- \htijb) \right]^2
\end{align}
subject to the single constraint 
\begin{align}
\label{singlec}
 \sum_i \eta_i \sum_j w_{ij} \thij = p.
\end{align}

In order to express (\ref{newloss}) and (\ref{singlec}) more compactly, we define
$\thhi = (\hat{\theta}_{i1},\ldots, \hat{\theta}_{in_{i}})^T,$\; 
$ \hat{\theta}_i^B = (\hat{\theta}_{i1}^B,\ldots, \hat{\theta}_{in_{i}}^B)^T,$\; 
$\tilde{\theta}_i = \thhi - \thib,$\; $A_i = \text{Diag} \{ 
\xi_{i1},\ldots, \xi_{in_i}
\},$ and
$w_i = (w_{i1},\ldots, w_{in_i})^T.$
Also, define $\Psi_i = A_i + \phi_i w_i w_i^T,$ and $s_i = \sum_j w_{ij}^2 \xi_{ij}^{-1} = w_i^TA_i^{-1} w_i.$ We also write $\thw = \sum_{ij} \eta_i w_i^T\thib.$
Then the constrained minimization problem can be rephrased as minimizing   
\begin{align}
\label{final1}
\sum_i \tilde{\theta}_i^T \Psi_i  \tilde{\theta}_i \quad 
 \text{ subject to the constraint } \quad
p^* := p -\thw = \sum_i \eta_i w_i^T  \tilde{\theta}_i.
\end{align}
The solution to the above problem is given in the following theorem:

\begin{theorem}
The minimizing solution of (\ref{final1}) is given by 
\begin{enumerate}
\item[(a)] $\hat{\theta}_i^{BM1} = 
\hat{\theta}_i^{B}  + \dfrac{(p-\thw)\eta_i(1+\phi_is_i)^{-1}A_i^{-1}w_i}{\sum_{k=1}^m\eta_k^2s_k(1+\phi_ks_k)^{-1}}$
for all $i,$ and
\item[(b)] $\hat{\delta}_{i}^{BM1} = 
w_i^T\hat{\theta}_i^{BM1} 
=
\bar{\hat{\theta}}_i^{B} + \dfrac{(p-\thw)\eta_is_i(1+\phi_is_i)^{-1}}{\sum_{k=1}^m\eta_k^2s_k(1+\phi_ks_k)^{-1}}$ for all $i.$

\end{enumerate}
\end{theorem}
\begin{proof}
We will show that $\tilde{\theta}_i^{BM1} = \hat{\theta}_i^{BM1} - \hat{\theta}_i^{B}$
is the unique solution of (\ref{final1}). The proof of (b) immediately follows. 
To prove (a), we rewrite
\begin{align}
\label{three}
\sum_i \til^T \Psi_i \til
&=\sum_i (\til - \tilde{\theta}_i^{BM1} )^T \Psi_i (\til - \tilde{\theta}_i^{BM1} ) \notag \\
&\quad + \sum_i (\tilde{\theta}_i^{BM1})^T \Psi_i  \tilde{\theta}_i^{BM1}
+ 2\sum_i (\tilde{\theta}_i^{BM1})^T \Psi_i (\til - \tilde{\theta}_i^{BM1} ).
\end{align}
Showing the cross term in (\ref{three}) is zero will prove that
 $\tilde{\theta}_i^{BM1}$ is the unique solution of~(\ref{final1}). 
Let $q = \sum_{k=1}^m\eta_k^2s_k(1+\phi_ks_k)^{-1}$ and note that $\Psi_i A_i^{-1} w_i =  (1+ \phi_i s_i)w_i$ for all $i.$ 
We can now express
$\hat{\theta}_i^{BM1} = 
\hat{\theta}_i^{B}  + p^*q^{-1} \;\eta_i(1+\phi_is_i)^{-1}A_i^{-1}w_i.$
 Using the new expression of $\hat{\theta}_i^{BM1}$ and the new notation, we find
$$\Psi_i \tilde{\theta}_i^{BM1}
=\Psi_i (\hat{\theta}_i^{BM1} - \thib)
= \frac{p^*}{q} \eta_i (1+\phi_is_i)^{-1}
\Psi_i A_i^{-1} w_i
= \frac{p^*}{q} \eta_i 
w_i.
$$
This implies that 
$\sum_i (\tilde{\theta_i}^{BM1})^T \Psi_i(\tilb - \tilde{\theta}_i^{BM1})  = \dfrac{p^*}{q}(p^* - p^*) = 0,$ which follows from the identity given above and the constraint in (\ref{final1}). Then the cross product term in~(\ref{three}) is zero, which proves the result.
\end{proof}

\textbf{Remark 1.} The result simplifies somewhat when $\xi_{ij} = w_{ij}$ 
for all $i,j.$ Then $s_i = 1$ for all $i$ so that
\begin{align*}
\hat{\theta}_{ij}^{BM1} &= \htijb + \dfrac{(p-\thw)\eta_i(1+\phi_i)^{-1}}{\sum_{k=1}^m\eta_k^2(1+\phi_k)^{-1}},\\
\hat{\delta}_i^{BM1} &= \thbiw + \dfrac{(p-\thw)\eta_i(1+\phi_i)^{-1}}{\sum_{k=1}^m\eta_k^2(1+\phi_k)^{-1}}.
\end{align*}
It can be seen that all individuals $j$ in area $i$ receive the same adjustment. Further simplification is possible when $\phi_i = \eta_i \;(i=1,\ldots,m).$ \\

\textbf{Remark 2.} Suppose $\hat{\theta}_{ij}^B > 0$ for all $i$ and $j.$ 
%
%
If we choose $\xi_{ij} = w_{ij} (\hat{\theta}_{ij}^B)^{-1}$ and $\phi_i = (g\eta_i-1)(\bar{\hat{\theta}}_{iw}^B)^{-1},$ where $g > \max_{1 \leq i \leq m}{\eta_i}^{-1},$
then our benchmarked Bayes estimator in Theorem 1 reduces to the raked estimator at both levels. Recall that the raked estimators take the form 
\begin{equation*}
\hat{\theta}_i^R =
p \, \thbiw  (\bar{\hat{\theta}}_w^B)^{-1}\; \text{  and  }\;
\hat{ \theta}_{ij}^R = p \,\hat{\theta}_{ij}^B (\bar{\hat{\theta}}_w^B)^{-1},
\end{equation*}
where $\hat{\theta}_i^R$ represents the raked estimator at the area level and 
$\hat{\theta}_{ij}^R$ represents the raked estimator at the unit level.
Using the values for $\xi_{ij}$ and $\phi_i$ as above, we find that 
$s_i = \thbiw.$ Then
\[
\hat{\delta}_i^{BM1}=\thbiw + \frac{(p-\thw)\eta_i\thbiw\{g\eta_i\}^{-1}}{\sum_k \eta_k^2  \bar{\hat{\theta}}_{kw}^B\{g\eta_k \}^{-1}}
=p\thbiw (\thw)^{-1}=\hat{\theta}^R_i.
\]
Also,  
\[
\thij = \htijb + \frac{(p-\thw)\eta_i\{g\eta_i\}^{-1}w_{ij}w_{ij}^{-1}\htijb }
{\sum_k \eta_k^2  \bar{\hat{\theta}}_{kw}^B\{g\eta_k \}^{-1}} 
=p \; \htijb (\thw)^{-1} = \hat{\theta}^R_{ij}.
\]
Hence, we have shown that the raked estimators are  special cases of Theorem 1. \\

\textbf{Remark 3}: The posterior mean squared error (PMSE) of 
$\hat{\theta}_{ij}^{BM1}$ is given by 
$$E[(\thij - \hat{\theta}_{ij}^{BM1})^2 \mid \text{data}]
= V(\tij \mid \text{data}) +
(\thijb - \hat{\theta}_{ij}^{BM1})^2.
$$
Similarly, $\text{PMSE}(\hat{\delta}_i^{BM1})
= 
V(\de \mid \text{data}) +
(\hat{\delta}_i^B - \hat{\delta}_{i}^{BM1})^2,
$
where 
$\hat{\delta}_i^B$ is the Bayes estimator of $\de.$
It is now clear that no benchmarked estimator can outperform any other for all $i$ and $j$.  Essentially, the benchmarking constraint forces a certain ``overall'' adjustment to be made, and the way in which a particular benchmarked estimator allocates this adjustment across the areas and/or units determines how it performs for each area and/or unit compared to other benchmarked estimators.  No sensible benchmarked estimator can make uniformly smaller adjustments than another (without violating the benchmarking constraints), hence no sensible benchmarked estimator can perform better than another for every area and/or unit.

\subsection{Benchmarking a Weighted Variability at the Unit Level}

Often, in addition to controlling the averages as in Theorem 1, we want to control the variability of the estimates as well. Specifically, we may wish to constrain 
the weighted variability of the estimates at the unit level
 (in the $i$th small area equal to $h_i$), which is
given externally from administrative data or taken to be $h_i = \sum_{ij} w_{ij} E[(\tij -\bar{\theta}_{iw})^2|\thij].$
The following theorem provides a partial answer to this problem when $\xi_{ij} = w_{ij}$ for all $i$ and $j.$

\begin{theorem}
The minimizing solution of $E[\sum_i\sum_j w_{ij}(\thij-\tij)^2 + \sum_i\phi_i(\de-\tbiw)^2|\text{data}]$ 
subject to (i)~$\sum_j w_{ij} \thij = \de,$
(ii)~$\sum_j w_{ij}(\thij-\de)^2 = h_i,$
and (iii)~$\sum_i \eta_i \de = p$
is given by
\begin{enumerate}
\item[(a)] $\hat{\theta}_{ij}^{BM2} = \thbiw  + \dfrac{(p-\thw)\eta_i(1+\phi_i)^{-1}}{\sum_k\eta_k^2(1+\phi_k)^{-1}} + \left(h_id_i^{-1}\right)^{\frac{1}{2}}
(\htijb-\thbiw)$
\item [(b)]$\hat{\delta}_i^{BM2} = \thbiw + \dfrac{(p-\thw)\eta_i(1+\phi_i)^{-1}}
{\sum_k\eta_k^2(1+\phi_k)^{-1}}$,
where $d_i= \sum_j w_{ij}(\htijb - \thbiw)^2.$
\end{enumerate}
\end{theorem}

\begin{proof}
Similar to the proof of Theorem~1, the problem reduces to the minimization of 
\begin{align}
\label{newloss2}
\sum_i \sum_j \xi_{ij}(\thij-\htijb)^2 
+\sum_i\phi_i\left[\sum_jw_{ij} (\thij- \htijb) \right]^2 
\end{align}
subject to constraints (i)--(iii).
Since $\xi_{ij} = w_{ij}$ for all $i$ and $j,$ we define $W_i = A_i$ for all~$i.$
We use the same notation as in the proof of Theorem~1, namely, $w_i, \hat{\theta}_i,
\thib, \tiw, \thbiw,$ and  $\tbw.$  
We also introduce additional notation, defining $1_i$ as an $n_i$ by $1$ 
column vector with each element equal to 1. 
We can now re-express (\ref{newloss2}) as
\begin{align}
\label{ident2}
& \sum_i (\thhi - \thib)^T(W_i + \phi_i w_i w_i^T)(\thhi - \thib) \notag\\
&\quad = \sum_i \left[(\thhi - \tiw1_i) - (\thib - \thbiw 1_i)\right]^T
(W_i + \phi_i w_i w_i^T)
\left[(\thhi - \tiw1_i) - (\thib - \thbiw 1_i)\right] \notag\\
&\quad + \sum_i  (\tiw - \thbiw)^T 
   1_i^T(W_i + \phi_i w_i w_i^T)
   1_i
(\tiw - \thbiw)
\notag \\
&\quad = \sum_i \left[(\thhi - \tiw1_i) - (\thib - \thbiw 1_i)\right]^T
(W_i + \phi_i w_i w_i^T)
\left[(\thhi - \tiw1_i) - (\thib - \thbiw 1_i)\right] \notag\\
&\quad + \sum_i (1+  \phi_i) (\tiw - \thbiw)^2,
\end{align}
since $w_i^T1_i = 1$ and $ W_i1_i = w_i.$
Moreover, combining constraints (i) and (iii) into one single constraint yields
(i') $\sum_i \eta_i w_i^T \thhi~=~p.$ The problem reduces to minimization of~(\ref{ident2})
with respect to $\that$ subject to (i') and (ii).
Now define $$g_i 
=  \dfrac{(p - \thw )\eta_i (1+ \phi_i)^{-1}}
{\sum_s \eta_s^2 (1 + \phi_s)^{-1}}$$ for $i=1,\ldots,m.$
Then the second term in the right-hand side  of (\ref{ident2}) becomes
\begin{align}
\label{gi}
\sum_i (1+  \phi_i) (\tiw - \thbiw)^2
=
\sum_i (1+  \phi_i) (\tiw - \thbiw - g_i)^2
+ \frac{(p - \thw)^2}{\sum_s \eta_s^2 (1 + \phi_s)^{-1}}.
\end{align}

Hence, the left-hand side of (\ref{gi}) is minimized when 
$\tiw = \thbiw + g_i$ for all i.\\
We now note that constraint~(ii) may be written as 
\begin{align*}
h_i&=(\thhi - \tiw1_i)^TW_i(\thhi - \tiw1_i)
=\thhi^TW_i\thhi-\left(\tiw\right)^2
=\thhi^T\left(W_i-w_iw_i^T\right)\thhi.
\end{align*}
Also, express
$d_i = (\thib- \thbiw 1_i)^T W_i (\thib- \thbiw 1_i).$ We reformulate the first expression in the right-hand side of (\ref{ident2}):
\begin{align}
\label{cauchy}
&\sum_i \left[(\thhi - \tiw1_i) - (\thib - \thbiw 1_i)\right]^T
(W_i + \phi_i w_i w_i^T)
\left[(\thhi - \tiw1_i) - (\thib - \thbiw 1_i)\right] \notag \\
&\quad =
\sum_i \left[(\thhi - \tiw1_i) - (\thib - \thbiw 1_i)\right]^T
W_i
\left[(\thhi - \tiw1_i) - (\thib - \thbiw 1_i)\right]\notag\\
&\quad=\sum_i(h_i+d_i)-2\sum_i(\thhi - \tiw1_i)^TW_i(\thib - \thbiw 1_i),
\end{align}
noting that $w_i^T(\thhi-\tiw1_i)=w_i^T(\thib-\thbiw)1_i=0$.  
Then since $h_i$ and $d_i$ do not depend on the choice of the benchmarked estimator~$\hat{\theta}$, minimization of (\ref{cauchy}) is equivalent to maximization of the expression
\begin{align}
\label{dagger}
\sum_i(\thhi - \tiw1_i)^TW_i(\thib - \thbiw 1_i).
\end{align}
Now define a pair of discrete random variables $Q_{1i}$ and $Q_{2i}$ with joint distribution $P(Q_{1i}=\hat{\theta}_{ij}-\bar{\hat{\theta}}_{iw},\;Q_{2i}=\hat{\theta}_{ij}^B-\bar{\hat{\theta}}_{iw}^B)=w_{ij}$.  Note that since $E(Q_{1i})=E(Q_{2i})=0$, the expression (\ref{dagger}) is equal to $\text{Cov}(Q_{1i},Q_{2i})$, which is maximized when $Q_{1i}$ and $Q_{2i}$ have correlation 1, i.e.,  when $\hat{\theta}_{i}-\bar{\hat{\theta}}_{iw}1_i=c_i(\hat{\theta}_{i}^B-\bar{\hat{\theta}}_{iw}^B1_i)$ for some positive scalar $c_i$.
This leads to 
$$h_i = (\thhi - \tiw1_i)^T W_i  (\thhi - \tiw1_i)
= c_i^2 (\thib - \thw 1_i)^T W_i 
(\thib - \thw 1_i) = c_i^2 d_i.
$$
Hence, $c_i = (h_i d_i^{-1})^{1/2}.$
Then
$
(\thhi - \tiw1_i) = (h_i d_i^{-1})^{1/2} (\thib - \thbiw 1_i). 
$
Combining this with $\tiw = \thbiw + g_i,$ we get $\hat{\theta}_{ij}^{BM3}$ as the minimizer of (\ref{ident2}). This proves part (a) of the theorem, and part (b) immediately follows since $\de = w_i^T \thhi.$
\end{proof}

\textbf{Remark 4.} We mention that the area-level estimates $\delta_i$ under Theorems 1 and 2 are identical under the special case when $\xi_{ij} = w_{ij},$ which  intuitively makes sense. Hence, in many situations where we are only interested in the area-level estimates for two-stage benchmarking, they will be the same under both theorems. However, the unit-level estimates $\hat{\theta}_{ij}$ will differ.

\subsection{Multiparameter Extensions}
\label{multi_parameter extentions}
We extend Theorem 1 to a multivariate setting. Consider the notation, where
$\theta = (\theta_{11},\ldots,\theta_{1n_1},\ldots,\theta_{m1},\ldots \theta_{mn_m})^T,$\;
$\hat{\theta} = (\hat{\theta}_{11},\ldots,\hat{\theta}_{1n_1},\ldots,\hat{\theta}_{m1},\ldots \hat{\theta}_{mn_m})^T,$ and
$\delta = (\delta_1,\ldots,\delta_m)^T;$ 
where $\theta_{ij}, \hat{\theta}_{ij}$, and $\delta_i$ are all vector-valued.
Define $\tbiw = \sum_j W_{ij}\tij$ and $\bar{\theta}_w = \sum_i\Gamma_i\tbiw,$ where 
$W_{ij}$ and $\Gamma_i$ are known matrices (for all $i$ and $j$). 
%
Let $\hat{\theta}^B = 
(\hat{\theta}_{11}^B,\cdots,\hat{\theta}_{1n_1}^B,\cdots,\hat{\theta}_{m1}^B,\cdots,\hat{\theta}_{mn_m}^B)^T$
denote the Bayes estimators.
We also define $\thbiw = \sum_j W_{ij}\htijb$ and $\thw = \sum_i \Gamma_i \thbiw.$
\begin{theorem}
Consider the loss function 
$$L(\theta,\that,\delta) =
\sum_i\sum_j (\thij - \tij)^T\Lambda_{ij}(\thij - \tij)
+ \sum_i (\de - \tbiw)^T\Psi_{i}(\de - \tbiw),$$ 
where $\Lambda_{ij}$ and $\Psi_i$ are positive definite matrices.
Then the minimizing solution of $E[L(\theta,\that,\delta)|\text{data}]$ subject to
(i) $\sum_jW_{ij}\thij = \de$ and \mbox{(ii) $\sum_i \Gamma_i \de = p$}
 is given by
\begin{itemize}
\item[(a)] $\hat{\theta}_{ij}^{BM3} = \htijb
+ \Lambda_{ij}^{-1}W_{ij}^Ts_i^{-1}(\Psi_i + s_i^{-1})^{-1}\Gamma_i^TR^{-1}(p - \thw)$
\item[(b)] $\hat{\delta}_i^{BM3} = \thbiw + (\Psi_i + s_i^{-1})^{-1}\Gamma_i^T R^{-1}(p - \thw),$
where $s_i = \sum_jW_{ij}\Lambda_{ij}^{-1}W_{ij}^T$ and \\
$R = \sum_i \Gamma_i(\Psi_i + s_i^{-1})^{-1}\Gamma_i^T.$
\end{itemize} 
\end{theorem}
Note that $\Lambda_{ij}$ need not be the same as $W_{ij}.$ Similarly, $\Psi_i$ need not be the same as $\Gamma_i.$ 
\begin{proof}
Similar to the proof of Theorem~1, the problem reduces to minimization of 
\begin{align}
\sum_{ij} (\thij - \thijb)^T \Lambda_{ij} (\thij - \thijb) +
\sum_i \left[ \sum_j
W_{ij} (\thij - \thijb)
\right] ^T\Psi_i
\left[ \sum_j
W_{ij} (\thij - \thijb)
\right] 
\end{align}
subject to the constraint $\sum_i \Gamma_i \sum_j W_{ij} \thij = p.$
We now define 
$\tilde{\theta}_{ij} = \hat{\theta}_{ij} - \thijb,$\;
$\thhi^T = (\hat{\theta}_{i1}^T,\ldots, \hat{\theta}_{in_i}^T),$\;
$(\thib)^T = ((\hat{\theta}_{i1}^B)^T,\ldots, (\hat{\theta}_{in_i}^B)^T),$\;
$\tilde{\theta}_i = \hat{\theta}_{i} - \thib,$ and
$A_i = \text{block diagonal}\{
\Lambda_{i1},\ldots, \Lambda_{in_i}
\}.$ In addition, let 
$W_i = 
(W_{i1},\ldots, W_{in_i})$,
so that $s_i = W_iA_i^{-1}W_i^T.$
Then we can reframe the problem as minimizing 
\begin{align}
\label{reframe}
\sum_i \tilde{\theta}_i^T (A_i + W_i^T\Psi_iW_i)\tilde{\theta}_i \quad\text{ subject to }
\sum_i \Gamma_iW_i \tilde{\theta}_i = p - \thw 
\end{align}
 with respect to $\tilde{\theta}_i.$
Write 
$\tilde{\theta}_i^{BM3} = \hat{\theta}_i^{BM3} - \thib.$ Then 
\begin{align}
\label{newloss3}
\sum_i \tilde{\theta}_i^T (A_i + W_i^T\Psi_iW_i)\tilde{\theta}_i
&= 
\sum_i (\tilde{\theta}_i  - \tilde{\theta}_i^{BM3} +  \tilde{\theta}_i^{BM3} 
)
^T (A_i + W_i^T\Psi_iW_i)
(\tilde{\theta}_i  - \tilde{\theta}_i^{BM3} +  \tilde{\theta}_i^{BM3}
)\notag \\
&=
\sum_i
(\tilde{\theta}_i  - \tilde{\theta}_i^{BM3} 
)
^T (A_i + W_i^T\Psi_iW_i)
(\tilde{\theta}_i  - \tilde{\theta}_i^{BM3} )\\
& + 
\sum_i
(  \tilde{\theta}_i^{BM3} 
)
^T (A_i + W_i^T\Psi_iW_i)
(\tilde{\theta}_i^{BM3}
)\notag\\
&+2 \sum_i
( \tilde{\theta}_i^{BM3} 
)
^T (A_i + W_i^T\Psi_iW_i)
(\tilde{\theta}_i  - \tilde{\theta}_i^{BM3} 
).\notag 
\end{align}
Since $s_i = W_i A_i^{-1} W_i^T,$ we find
\begin{align}
\label{ab}
(A_i + W_i^T\Psi_iW_i)
\tilde{\theta}_i^{BM3}\;
&=
(A_i + W_i^T\Psi_iW_i) A_i^{-1} W_i^T s_i^{-1} (\Psi_i + s_i^{-1})^{-1}
\Gamma_i^T R^{-1} (p - \thw)\notag \\
&=
W_i^T (\Psi_i + s_i^{-1})(\Psi_i + s_i^{-1})^{-1}
\Gamma_i^T R^{-1} (p - \thw)\notag\\
&= 
W_i^T \Gamma_i^T R^{-1} (p - \thw).
\end{align}
This leads to 
\begin{align}
\label{end}
\sum_i (
 \tilde{\theta}_i^{BM3} 
)
^T (A_i + W_i^T\Psi_iW_i)
(\tilde{\theta}_i  - \tilde{\theta}_i^{BM3} 
)
&=  (p - \thw)^{T}R^{-1}\sum_i \Gamma_i W_i (\tilde{\theta}_i -
\tilde{\theta}_i ^{BM3})\notag \\
&= 
(p - \thw)^{T}R^{-1} \left[
(p - \thw) - (p - \thw)
\right] = 0
\end{align}
due to the constraint on $\tilde{\theta}_i$ and the fact that
\begin{align*}
\sum_i \Gamma_i W_i \tilde{\theta}_i ^{BM3}
= \sum_i 
\left \{
\Gamma_i  (\Psi_i + A_i^{-1})^{-1}  \Gamma_i^T
\right\}
R^{-1} (p  - \thw) = p - \thw
\end{align*}
from the definition of $R.$
In view of (\ref{newloss3}) and (\ref{end}), the unique solution of 
(\ref{reframe}) is given by $\tilde{\theta}_i = \hat{\theta}_i^{BM3}.$ This proves part (a) of the theorem, and part (b) immediately follows.
\end{proof}

\section{Application to an Asian Subpopulation from the NHIS}
\label{app:nhis}
This section considers small area/domain estimation of the proportion of people without health insurance for several domains of an Asian subpopulation from the National Health Interview Survey (NHIS) in 2000. Our goal is to benchmark the aggregated probabilities that a person does not have health insurance to the corresponding weighted proportions in the sample for each domain. We also benchmark the weighted proportions in the sample for each domain to match the overall estimated population proportions. 

\subsection{Specifics on the Asian Subpopulation}

The 96 domains were constructed on the basis of age, sex, race, and the size of the  metropolitan statistical area in which the individual resides.
The NHIS study provides the individual-level binary response data, whether or not a person has health insurance, as well as the individual-level covariates,  where the covariates used are 
family size, education level, and total family income. 
The Asian group is made up of the following four race groups: Chinese, Filipino, Asian Indian, and others (such as Korean, Vietnamese, Japanese, Hawaiian, Samoan, Guamanian, etc.). Individuals from these subpopulations are assigned to specific domains depending on their age, gender, and region. There are three age groups (0--17, 18--64, 65+). Furthermore, there are two genders, four races, and four regions that depend on the size of the metropolitan statistical area ($\leq$~499,999; 500,000--999,999; 1,000,000--2,499,999; $\geq$~2,500,000). Hence, there are $4\times2\times4\times3 =96$ total domains. 
In \cite{ghosh_2009} a stepwise regression procedure is used for this data set to reach a final model with an intercept term and three covariates: family size, education level, and total family income. 
%
%

\subsection{A Hierarchical Bayesian Model}
When targeting specific subpopulations cross-classified by demographic 
characteristics, direct estimates are usually accompanied by large standard 
errors and coefficients of variation. Hence, a procedure such as the one 
proposed in Section 2 is appropriate. Let
 $N = \sum_i n_i,$ and $Z = \text{block diagonal}(1_{n_1}, \ldots, 1_{n_m})$. 
Define $y_{ij} $ to be a Bernoulli random variable with value 1 if a person does not have health insurance and 0 otherwise. We denote $p_{ij}=P(y_{ij}=1)$ and
$\theta_{ij}=\mbox{logit}(p_{ij})$. The model considered is given by
\begin{align}
\label{model}
y_{ij}|\theta_{ij}&\stackrel{ind}{\sim} \text{Bin}(1, \theta_{ij}),\\
\theta_{ij} |u_i&\stackrel{ind}{\sim} N(x_{ij}^T\beta+u_i,\sigma_e^2),\notag\\
u_i|\sigma_u^2&\stackrel{iid}{\sim}N(0, \sigma_u^2),\notag\\
\beta_{p \times 1}& \sim \text{Uniform}(\mathbb{R}^p),\notag\\
\sigma_e^2 &\sim \text{InverseGamma}(a/2, b/2),\notag\\
\sigma_u^2  &\sim \text{InverseGamma}(c/2, d/2),\notag
\end{align}
where $\beta$, $\sigma_e^2$, and $\sigma_u^2$ are
mutually independent. Define the notation as:
$X^T=(x_{11},\ldots,x_{1n_1},\ldots,x_{m1},\ldots,x_{mn_m})$,
$\theta^T=(\theta_{11},\ldots,\theta_{1n_1},\ldots,\theta_{m1},\ldots,
\theta_{mn_m})^T$, $u=(u_1,\ldots,u_m)^T$. We also define $$D = \text{Diag}\left(
\left(\dfrac{n_1}{\sigma_e^2} + \dfrac{1}{\sigma_u^2}\right)^{-1}, \ldots,
\left(\dfrac{n_m}{\sigma_e^2} + \dfrac{1}{\sigma_u^2}\right)^{-1}
\right).$$
We derive the full conditionals from \cite{ghosh_1998}, 
\begin{align*}
\beta \mid \theta, u, \sigma_e^2, \sigma_u^2, y &\sim
 N\left(\;(X^TX)^{-1} X^T(\theta - Zu), \; \sigma_e^2(X^TX)^{-1} \right),\\
u \mid \theta, \beta, \sigma_e^2, \sigma_u^2, y &\sim 
N\left(\;
\frac{1}{\sigma_e^2}DZ^T(\theta-X\beta),\;D
\; \right),\\
\sigma_e^2 | \theta, u, \beta, \sigma_u^2, y & \sim
\text{InverseGamma}\left(\frac{1}{2}\left[
a + (\theta - X\beta - Zu)^T(\theta - X\beta - Zu)
\right],\;
\frac{1}{2}(b + N)
\right),\\
\sigma_u^2 | \theta, u, \beta, \sigma_e^2, y & \sim
\text{InverseGamma}\left(\frac{1}{2}\left(
c + u^Tu
\right),\;
\frac{1}{2}(d + m)
\right),\\
\pi(\theta_{ij} \mid u, \beta,\sigma_u^2, \sigma_e^2, y  )& \propto
\exp\left[
y_{ij} \theta_{ij} - \log(1 + e^{\theta_{ij}})
-\frac{(\theta_{ij} - x_{ij}^T\beta- u_i)^2}{2\sigma_e^2}
\right].
\end{align*}

We implement a Gibbs sampler with the $\theta_{ij}$ drawn by an embedded rejection sampling scheme, which has 200,000 Gibbs iterations, a burn-in of 2000, and a thin value of 200. Various diagnostics were checked such as trace plots, autocorrelation plots, and others until we were satisfied that the chain was not failing to converge (using two different starting values for the parameters being monitored). 

\subsection{Analyzing the Data}
In our analysis, we first find the Bayes estimates and associated standard errors (of the small domains) using Markov chain Monte Carlo (MCMC). Let $\hat{\theta}_{ij}^{(m)}$ denote the sampled value of $\theta_{ij}$ from the MCMC output generated from the $m$th draw, where there are $M$ total draws. The Monte Carlo estimate of $E(\tij|\textbf{y}) = M^{-1}\sum_{m=1}^M \hat{\theta}_{ij}^{(m)}.$ The Monte Carlo estimate of 
$$\text{Cov}(\tij, \theta_{i'j'}|\textbf{y}) = 
M^{-1}\sum_{m=1}^M(\hat{\theta}_{ij}^{(m)}\hat{\theta}_{i'j'}^{(m)})
- (M^{-1}\sum_{m=1}^M\hat{\theta}_{ij}^{(m)})(M^{-1}\sum_{m=1}^M\hat{\theta}_{i'j'}^{(m)}).$$ 
Then 
\begin{equation*}
E(\bar{\theta}_{iw}|\textbf{y}) = \sum_{j=1}^{n_i}w_{ij}E(\tij|\textbf{y})\; \text{ and }\; \text{Var}(\bar{\theta}_{iw}|\textbf{y}) 
= \sum_{j=1}^{n_i}w_{ij}w_{ij'}\text{Cov}(\tij,\theta_{ij'}|\textbf{y}).
\end{equation*}
 We then find the two-stage benchmarked Bayes estimates using Theorem 1 with the weights to be discussed below. Furthermore, to measure the variability associated with the two-stage  benchmarked Bayes estimates at the area levels, we use 
the posterior mean squared error (PMSE) under the assumed Bayesian model
as given in Remark 3. This is given by
\begin{align}
\text{PMSE}(\hat{\delta}_i^{BM1}) = 
 (\hat{\delta}_i^{B} - \thbariw)^2 + \text{Var}(\bar{\theta}_{iw}|\textbf{y}).
\end{align}
%
We consider the following constraint weights, where $w_{ij}^* =$ final person weight in the NHIS dataset:
\begin{equation*}
w_{ij} = \frac{w_{ij}^{*}}{\sum_j w_{ij}^{*}} \; \text{ and } \eta_{i} = \frac{\sum_jw_{ij}^*}{\sum_i\sum_jw_{ij}^*}.
\end{equation*}
 This leads to $p = \frac{\sum_i \sum_jw_{ij}^{*}y_{ij}}{\sum_i \sum_j w_{ij}^{*}} = 
\sum_i \sum_j (w_{ij}  y_{ij})\eta_i.$ 
In choosing the loss function weights, we consider three different sets of weights that will define three different estimators: the inverse-variance estimator~$\di$, the raked estimator~$\dr,$ and the domain-weighted inverse-variance estimator~$\dn.$
We also consider a na\"{i}ve estimator~$\dc$ for comparison purposes.
Each estimator and its corresponding loss function weights are described below, where the constraint weights for benchmarking are to be held fixed.
\begin{enumerate}
\item
We first consider the constant choice $\xi_{ij}=\phi_i=1$ for all $i$ and $j$, which leads to what we call the constant-weight estimator~$\dc$.  It will be seen later that this choice yields undesirable results.
\item The
second
set of weights corresponds to what we call the inverse-variance estimator $\hat{\delta}_i^I.$ To achieve this estimator we take $\xi_{ij} = \V^{-1}(\theta_{ij}|\textbf{y})$ and $ \phi_i = \V^{-1}(\sum_j w_{ij} \theta_{ij}|\textbf{y}).$
We choose this set of weights via the suggestion of \cite{isaki_2004}. Intuitively, if the Bayes estimator is thought to be a good estimator, then the posterior variance of $\tij$ will be small, meaning that $\xi_{ij}$ will be large. 
This will assign greater weight to these small areas, meaning that the loss penalizes errors in these domains more heavily.
Similarly, the same behavior will hold true for our choice of $\phi_i.$ The opposite behavior will happen in the reverse scenario. 
\item The
third
set of weights we take is given in Remark 2 of Theorem 1, which results in the raked estimator $\hat{\delta}_i^R.$ 
\item Finally, note that the weights suggested by \cite{isaki_2004} do not take into account the dramatically different sample sizes of the domains.  To remedy this, we define the domain-weighted inverse-variance estimator~$\dn$ to have the following weights: $\xi_{ij} =  V^{-1}(\theta_{ij}| \bm{y})$ and $\phi_i=  n_i
V^{-1}(\sum_jw_{ij}\theta_{ij}| \bm{y}).$ 
\end{enumerate}
An overall general simplification of the area-level estimates can be found by plugging into Theorem~1. Moreover, the area-level estimates under Theorem~2 are identical to those under Theorem~1.  

\subsection{Comparison of Benchmarked Bayes Estimators}
In this section, we illustrate that the four estimators considered in Section 3.3 never dominate each other in terms of posterior mean squared error, which is consistent with our theoretical results. 
(See the appendix for all figures).\\

The left-hand plot in Figure~1 illustrates the undesirable behavior of the constant-weight estimator~$\dc.$  The constant-weight estimator~$\dc$ tends to make substantially smaller adjustments in domains with smaller sample sizes, which is difficult to logically defend.  This effect is reduced under the inverse-variance estimator~$\di$, but it is not eliminated entirely.  The domain-weighted inverse-variance estimator~$\dn$ takes this idea a step further by the introduction of an extra factor that more severely penalizes large adjustments in areas with large sample sizes.  The middle plot of Figure~1 demonstrates that~$\dn$ indeed has the effect intended, as the adjustments tend to be smaller in domains with larger sample sizes.  The right-hand plot compares the domain-weighted inverse-variance estimator~$\dn$ to the raked estimator~$\dr.$ The magnitude of the adjustment made by the raked estimator~$\dr$ depends only on the value of the Bayes estimate, not on the domain sample size, which differs markedly from the behavior of the domain-weighted inverse-variance estimator~$\dn$.\\

An important feature of the plots in Figure~1 is that none of the four estimators can be said to dominate any of the others in terms of posterior mean squared error, although we were able to identify the behavior of the constant-weight estimator~$\dc$ as undesirable for other reasons.  (See Remark~3 for an explanation of these estimators' inability to dominate each other.)  Thus, it is not immediately clear which of $\di$, $\dn$, or $\dr$ would be the most desirable set of estimates to report in this situation.  However, what these results demonstrate is that the constrained Bayesian approach to the benchmarking problem is quite flexible, as different choices of the loss function weights can lead to estimators that exhibit a wide variety of behavior.  The Bayesian benchmarking framework provides a sound theoretical justification for these various estimators and unifies them as special cases of the much more general result in Theorem~1.\\

Figure 2 shows the three given estimators: the raked estimator $\dr$, the inverse-variance estimator  $\di$, and the domain-weighted inverse-variance estimator $\dn.$ Each estimator's percent increase in posterior root mean square error (\%~PRMSE) is plotted versus the domain sample size ($n_i$). We observe that the raked estimator has the least amount of spread for the smaller domains, while the domain-weighted estimator has less spread for larger domains. The domain-weighted estimator has the most consistent overall spread as the domain sample size increases, which makes sense since this estimator weights the loss function partially accordingly to the domain sample size $n_i.$


\subsection{Simulation Study}
We perform a simulation study to address how the model performs under simulated data compared to the model proposed in (\ref{model}). For simulation, we use the same covariates from the NHIS study and we use the posterior estimates of $\beta,$ $\theta$, $\sigma_u^2$ and~$\sigma_e^2$ from the run performed in the previous section. We draw a simulated value~$y_{\text{sim}}$ according to the following model:
\begin{align}
\label{sim}
y_{{\text{sim}}_{N \times 1}} \mid \theta_{N \times 1}  &\stackrel{ind}{\sim} \text{Bin}(1, \theta), \\
\text{logit}(\theta_{N \times 1} ) &= X\beta + Zu + e,\notag\\
u_{m \times 1} \mid  \sigma_u^2 & \stackrel{iid}{\sim} N(0,  \sigma_u^2 I_m),\notag\\
e_{N \times 1} \mid  \sigma_e^2 & \stackrel{iid}{\sim} N(0,  \sigma_e^2 I_N)\notag
\end{align}
with known parameter values coming from the posterior estimates of $\beta,$ $\theta$, $\sigma_u^2$ and $\sigma_e^2$  already mentioned. We input the simulated $y_{\text{sim}}$ into our Gibbs sampler based on 
the model in (\ref{model}). \\


In Figure 3, we compare the difference of the proposed estimators and the Bayes estimator versus the domain sample size. 
On the left, we plot the difference $\dr - \db$ versus $n_i$ for the simulated data and NHIS data. We also plot the difference $\di - \db$ versus $n_i$ for the simulated and NHIS data which corresponds to the middle plot. We plot the same for $\dn-\db$ on the right.
All three plots illustrate similar behavior regardless of the data source. \\

In addition, the left-hand plot of Figure~\ref{plot_est} plots the \% increase in PRMSE of the raked estimator under the NHIS data  and the simulated data versus the domain sample size. This plot illustrates that under both data sources there is more spread as the inverse-variance estimator increases. 
The middle figure plots the \% increase in PRMSE of the inverse-variance estimator under the NHIS and the simulated data versus the domain sample size, while the right-hand plot does the same for the domain-weighted inverse-variance estimator.
 All three plots illustrate that regardless of the source of the data, 
 there is more spread as the respective estimator increases. \\

Furthermore, we compared kernel density estimates of the raked estimator $\dr,$ the 
inverse-variance estimator $\di,$ and the domain-weighted inverse-variance estimator $\dn$ under the simulated model and the proposed model. That is, we looked at the difference between the respective estimators for the simulated and NHIS data, and estimated the kernel density of this quantity. Figure~\ref{plot_density} shows no irregular behavior except at the tails, which is not unexpected.\\

As an additional remark for general applications, often in small area problems, some local areas are not sampled. This is true,
for example, when local areas are counties. However, one can still obtain
hierarchical Bayes estimators for these areas, somewhat akin to the
regression estimators, by borrowing strength from other areas. Once these
Bayes estimators are obtained, the benchmarking procedures proposed in this
paper will apply, without posing any additional problem.

\section{Discussion}
A novel two-stage benchmarking procedure has been developed involving a general loss and no specific distributional assumptions. We have shown that it is possible to control both the mean and variability of the estimates. Doing so at the area level should be the next step for future research, proving to be a much more complex problem. 
We have illustrated the methodology using NHIS data as well as under a simulated scenario and shown that the raked, inverse-variance, and domain-weighted inverse-variance estimators do not dominate any other estimator in terms of posterior mean squared error, which agrees with our theoretical derivations. In terms of understanding complex surveys such as this one for the NHIS, there are many factors that need further exploration: how to best choose the weights, how to speed up algorithms for complex surveys to analyze real data and for simulation studies, and finally extensions of benchmarking to multi-stage procedures. \\

\textbf{Acknowledgements}:
This research was partially supported by the United States Census Bureau Dissertation Fellowship Program and NSF Grants SES 1026165  and SES 1130706. The views expressed reflect those of the authors and not of the National Health Interview Survey, the United States Census Bureau, or NSF. We would like to express our thanks to the Associate Editor and referees for their helpful suggestions.


\newpage
\bibliography{chomp}
%
%

\newpage
\section{Appendix}

\begin{figure}[htdp]
\begin{center}
\includegraphics[scale=0.6]{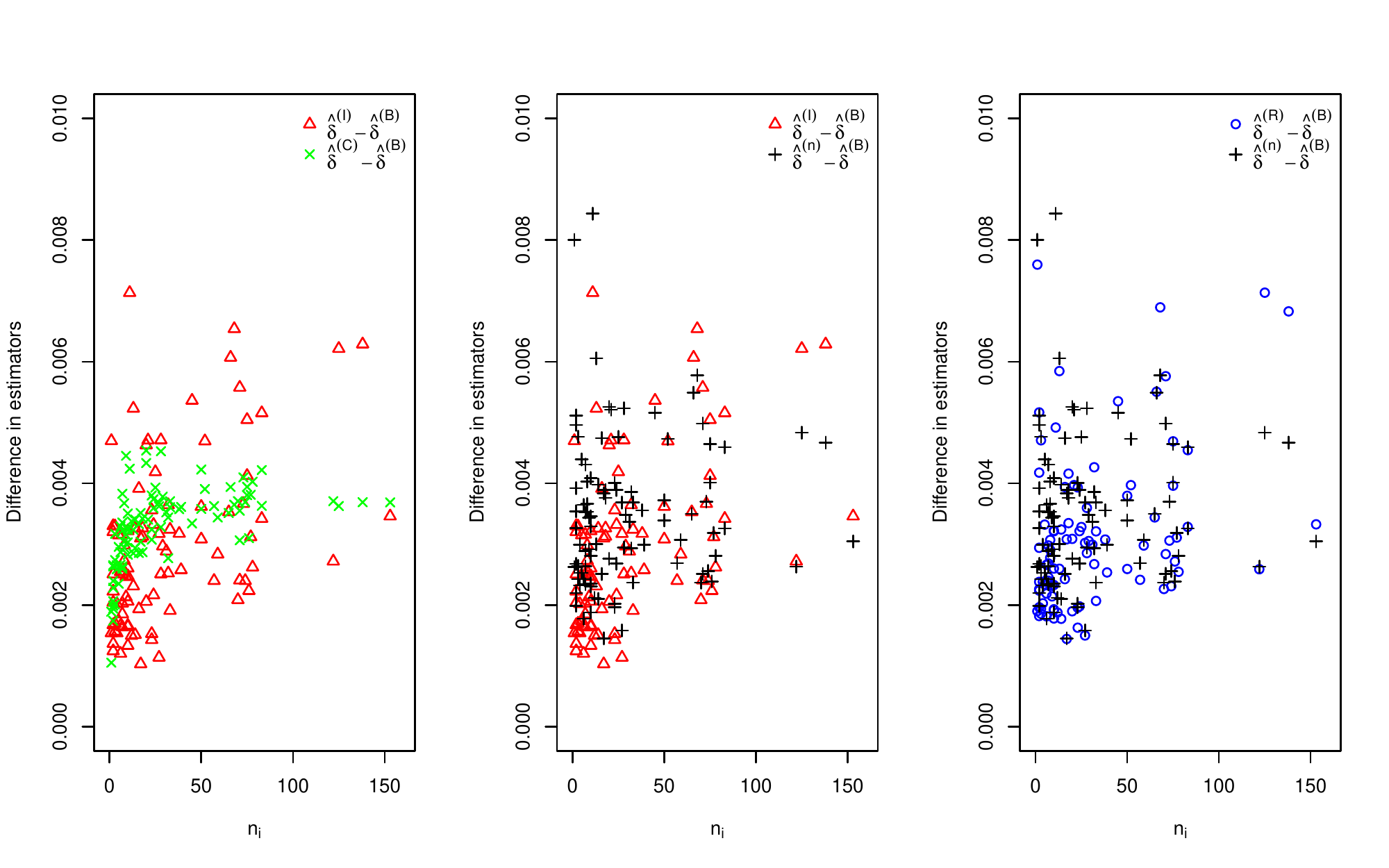}
\end{center}
\label{default}
\caption{The first plot shows $\di - \db$ and $\dc - \db$ versus domain sample size $n_i,$ illustrating that the estimator $\dc$ is not desirable. The middle plot shows $\di - \db$ and $\dn- \db$ versus domain sample size $n_i,$ while the last plot shows $\dr - \db$ and $\dn - \db$ versus domain sample size $n_i.$ 
}
\end{figure}%


\vspace*{-10em}
\begin{figure}[htdp]
\begin{center}
\includegraphics[scale=0.45]{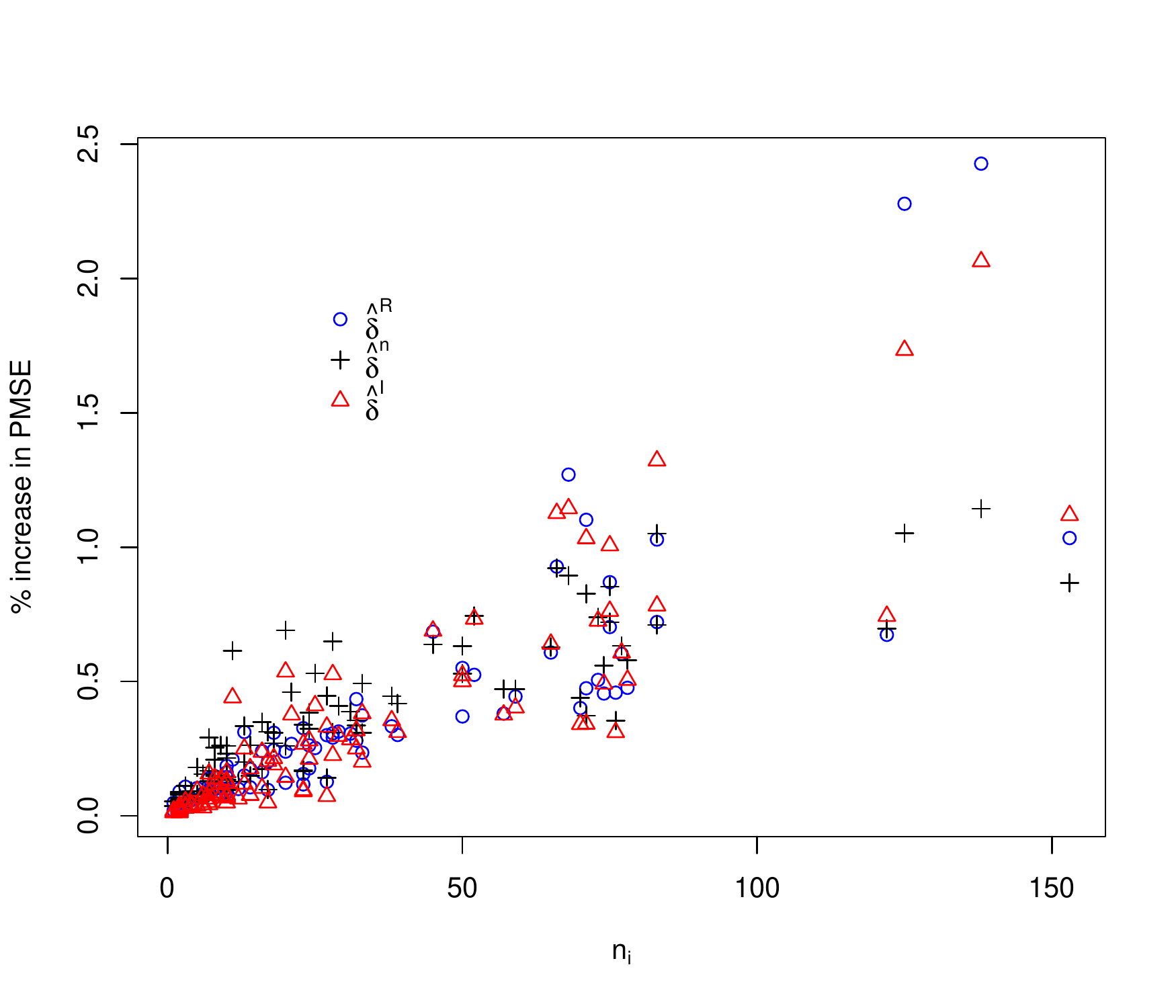}

\begin{tabular}{|c|c|}

\end{tabular}
\end{center}
\label{default}
\caption{This plot shows the percent increase in PMSE of the raked estimator  $\dr$, the inverse-variance estimator  $\di$, and the domain-weighted inverse-variance estimator $\dn$ versus the domain sample size $n_i.$ }
\end{figure}


\begin{figure}[h]
\center
\includegraphics[scale=.55]{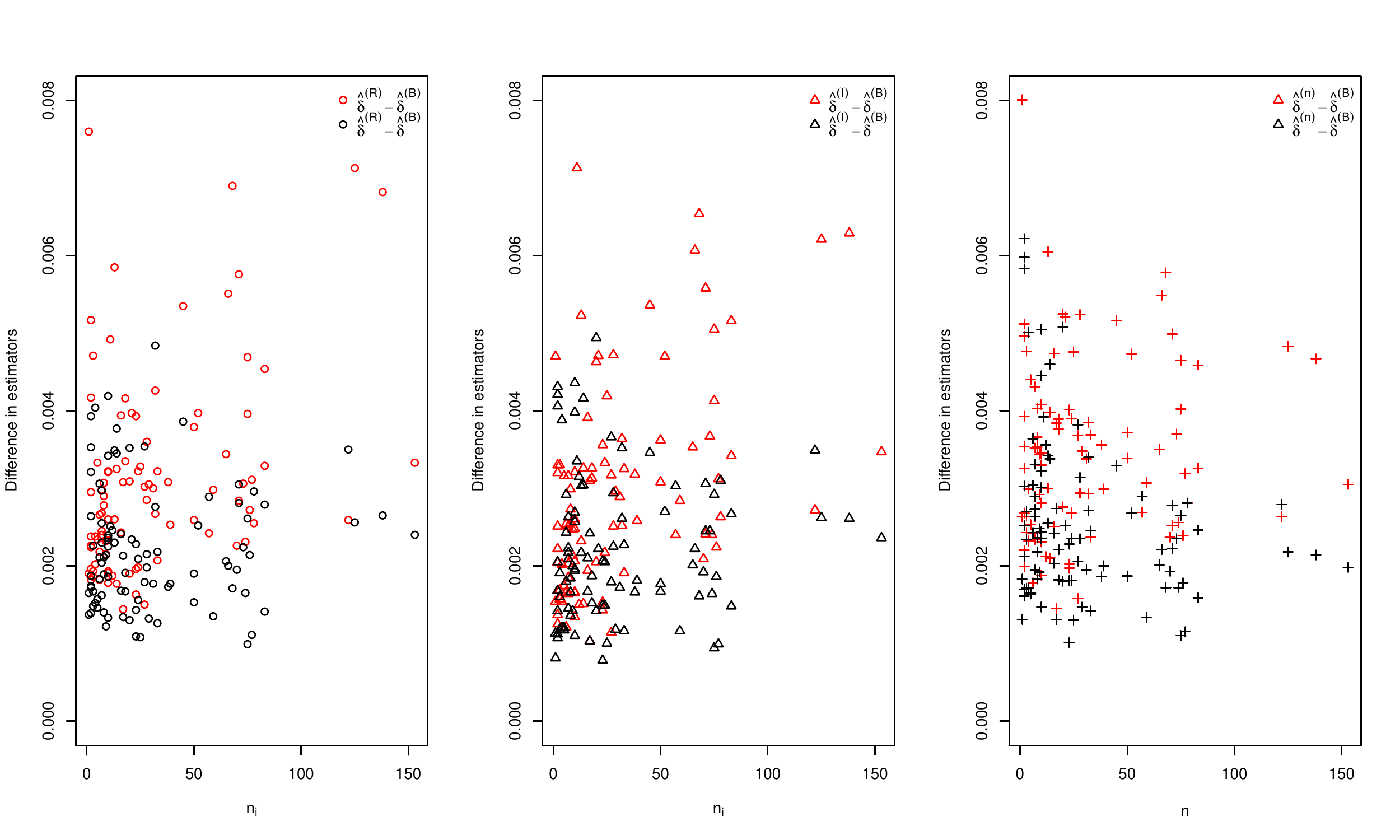}
\caption{On the left, we plot the difference $\dr - \db$ versus $n_i$ for the simulated and NHIS data. We also plot the difference $\di - \db$ versus $n_i$ for the simulated and NHIS data  which corresponds to the middle plot. We plot the same for $\dn-\db$ on the right.
All three plots illustrate similar behavior for the simulated data as the NHIS data.}
\label{plot_sim}
\end{figure}
\newpage

\begin{figure}[h]
\center
\includegraphics[scale=.5]{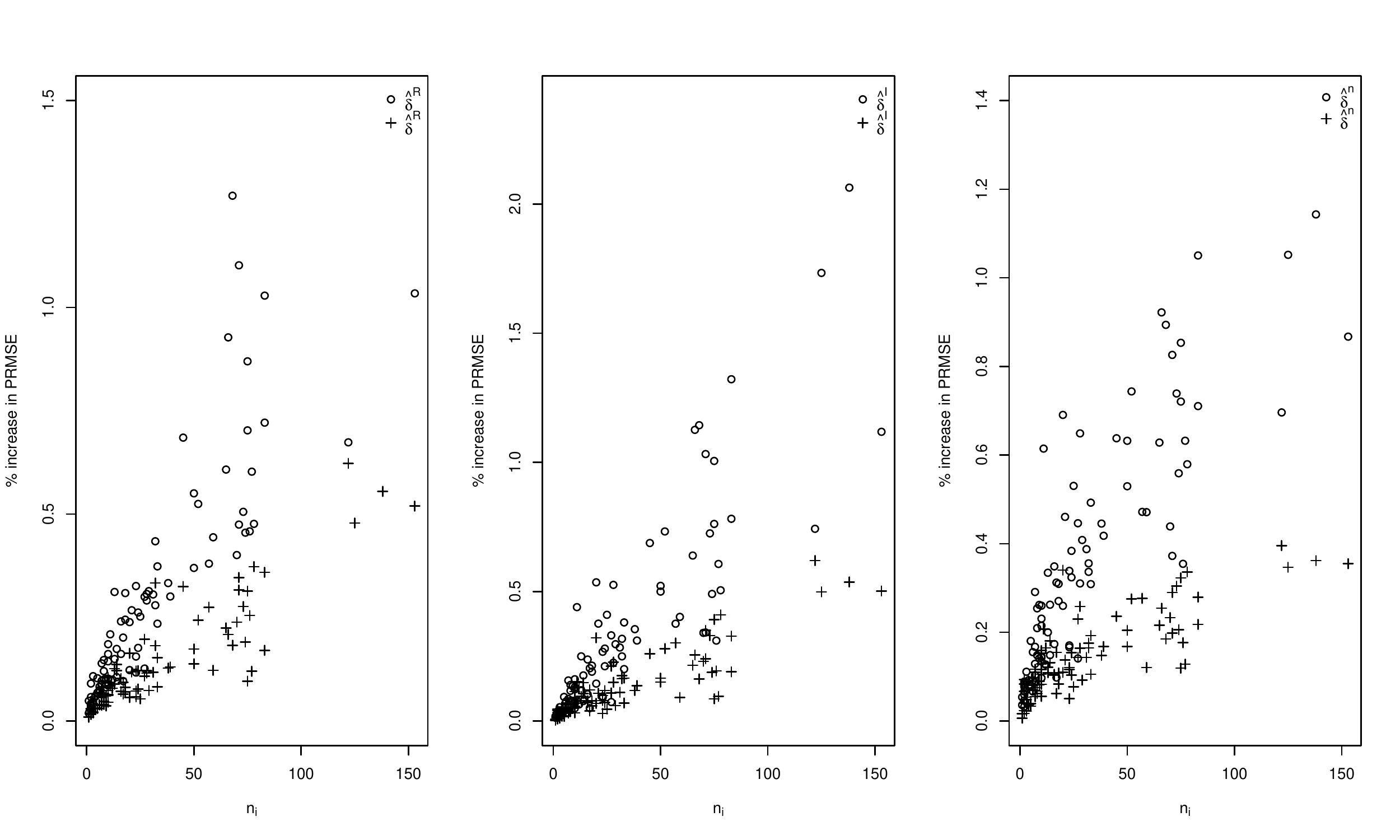}
\caption{The figure on the left plots the \% increase in PRMSE of the raked estimator under the NHIS and the simulated data versus the domain sample size. This plot illustrates that under both data sources there is more spread as the inverse-variance estimator increases. 
The middle figure plots the \% increase in PRMSE of the inverse-variance estimator under the NHIS data (circles) and the simulated data (plusses) versus the domain sample size, while the right-hand plot does the same for the domain-weighted inverse-variance estimator.}
\label{plot_est}
\end{figure}

\begin{figure}[h]
\begin{center}
\includegraphics[scale=.5]{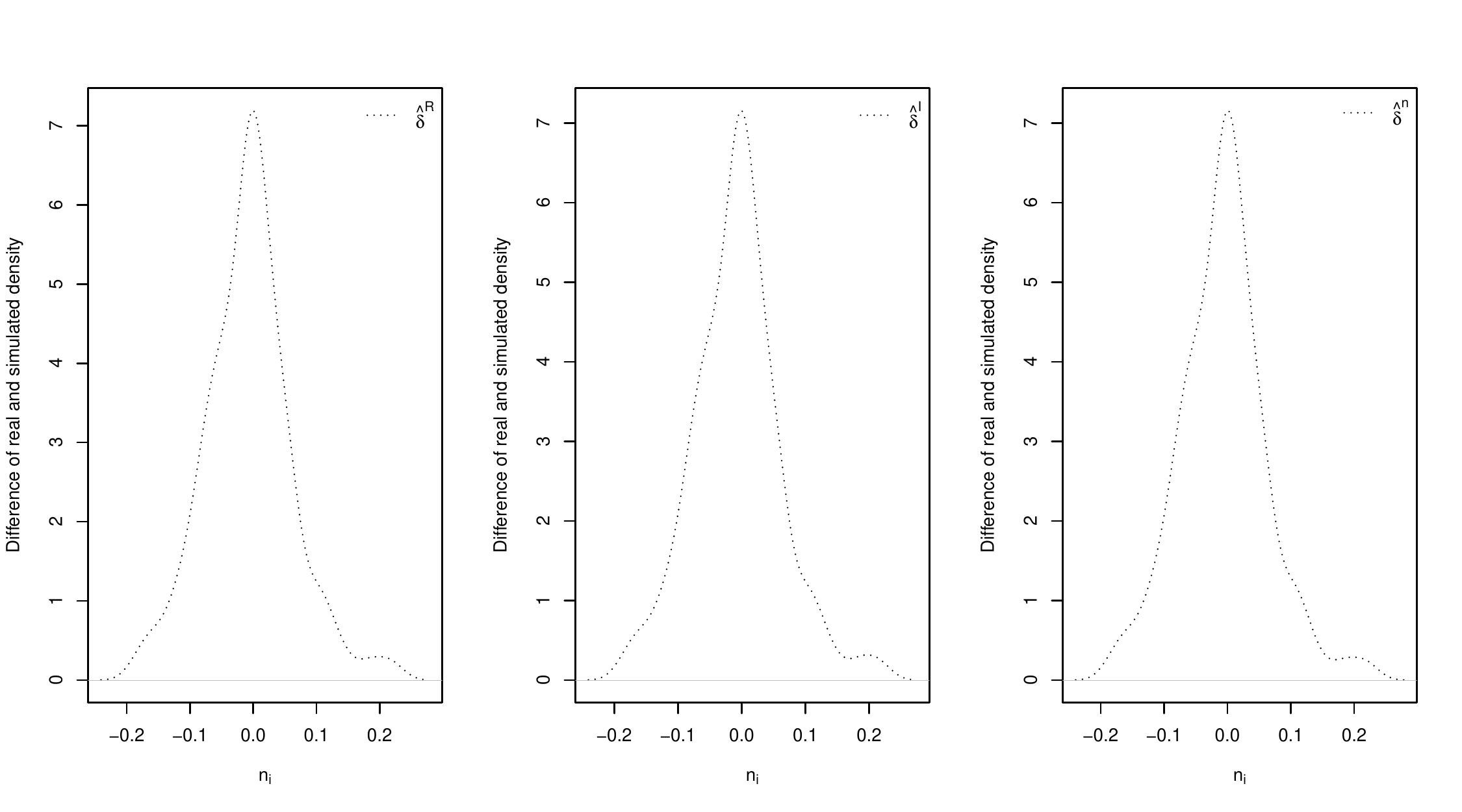}
\caption{The figure on the left plots a kernel density estimate of the difference of the inverse estimator calculated under (i) the NHIS data and (ii) the simulated data. The figure on the right does this for the raked estimator. Both plots illustrate that the posterior mean of the difference of the densities is centered near~0, with the function being unimodal.}
\label{plot_density}
\end{center}
\end{figure}

\end{document}